\documentclass[final]{IEEEtran}

\usepackage[mathscr]{eucal}
\usepackage[cmex10]{amsmath}
\usepackage{epsfig,epsf,psfrag}
\usepackage{amssymb,amsmath,amsthm,amsfonts,latexsym,bm}
\usepackage{amsmath,graphicx,bm,xcolor,url}
\usepackage[caption=false]{subfig} 
\usepackage{fixltx2e}
\usepackage{array}
\usepackage{verbatim}
\usepackage{bm}
\usepackage{algorithmic, cite}
\usepackage{algorithm}
\usepackage{verbatim}
\usepackage{textcomp}
\usepackage{mathrsfs}
\usepackage{multirow}
\usepackage{epstopdf}

\catcode`~=11 \def\UrlSpecials{\do\~{\kern -.15em\lower .7ex\hbox{~}\kern .04em}} \catcode`~=13 

\allowdisplaybreaks[1]

\newcommand{\calA}{\mathcal{A}}

\newcommand{\calE}{\mathcal{E}}

\newcommand{\calI}{\mathcal{I}}

\newcommand{\calQ}{\mathcal{Q}}

\newcommand{\calT}{\mathcal{T}}

\newcommand{\bb}{\mathbf{b}}

\newcommand{\bE}{\mathbf{E}}

\newcommand{\bX}{\mathbf{X}}
\newcommand{\by}{\mathbf{y}}
\newcommand{\bY}{\mathbf{Y}}

\newcommand{\rme}{\mathrm{e}}

\newcommand{\bbE}{\mathbb{E}}

\newcommand{\bbP}{\mathbb{P}}

\newcommand{\bbR}{\mathbb{R}}

\DeclareMathAlphabet{\mathbsf}{OT1}{cmss}{bx}{n}
\DeclareMathAlphabet{\mathssf}{OT1}{cmss}{m}{sl}

\newcommand{\rvP}{\mathsf{P}}

\DeclareSymbolFont{bsfletters}{OT1}{cmss}{bx}{n}  
\DeclareSymbolFont{ssfletters}{OT1}{cmss}{m}{n}
\DeclareMathSymbol{\bsfGamma}{0}{bsfletters}{'000}
\DeclareMathSymbol{\ssfGamma}{0}{ssfletters}{'000}
\DeclareMathSymbol{\bsfDelta}{0}{bsfletters}{'001}
\DeclareMathSymbol{\ssfDelta}{0}{ssfletters}{'001}
\DeclareMathSymbol{\bsfTheta}{0}{bsfletters}{'002}
\DeclareMathSymbol{\ssfTheta}{0}{ssfletters}{'002}
\DeclareMathSymbol{\bsfLambda}{0}{bsfletters}{'003}
\DeclareMathSymbol{\ssfLambda}{0}{ssfletters}{'003}
\DeclareMathSymbol{\bsfXi}{0}{bsfletters}{'004}
\DeclareMathSymbol{\ssfXi}{0}{ssfletters}{'004}
\DeclareMathSymbol{\bsfPi}{0}{bsfletters}{'005}
\DeclareMathSymbol{\ssfPi}{0}{ssfletters}{'005}
\DeclareMathSymbol{\bsfSigma}{0}{bsfletters}{'006}
\DeclareMathSymbol{\ssfSigma}{0}{ssfletters}{'006}
\DeclareMathSymbol{\bsfUpsilon}{0}{bsfletters}{'007}
\DeclareMathSymbol{\ssfUpsilon}{0}{ssfletters}{'007}
\DeclareMathSymbol{\bsfPhi}{0}{bsfletters}{'010}
\DeclareMathSymbol{\ssfPhi}{0}{ssfletters}{'010}
\DeclareMathSymbol{\bsfPsi}{0}{bsfletters}{'011}
\DeclareMathSymbol{\ssfPsi}{0}{ssfletters}{'011}
\DeclareMathSymbol{\bsfOmega}{0}{bsfletters}{'012}
\DeclareMathSymbol{\ssfOmega}{0}{ssfletters}{'012}

\newcommand{\tilN}{\tilde{N}}

\newcommand{\hatS}{\hat{S}}

\newcommand{\hatY}{\hat{Y}}

\newcommand{\bark}{\bar{k}}

\DeclareMathOperator{\var}{\mathsf{Var}}

\newtheorem{theorem}{Theorem} 
\newtheorem{lemma}{Lemma}

\newtheorem{proposition}{Proposition}
\newtheorem{corollary}{Corollary}

\newcommand{\qednew}{\nobreak \ifvmode \relax \else
      \ifdim\lastskip<1.5em \hskip-\lastskip
      \hskip1.5em plus0em minus0.5em \fi \nobreak
      \vrule height0.75em width0.5em depth0.25em\fi}

\usepackage{xspace}
\usepackage[ colorlinks = true,
             linkcolor = blue,
             urlcolor  = blue,
             citecolor = red,
             anchorcolor = green,
]{hyperref}

\usepackage{cite}
\allowdisplaybreaks[1]
\usepackage{amssymb}
\usepackage{pifont}

\begin{document}
\title{On the All-Or-Nothing Behavior \\ of Bernoulli Group Testing} 

\author{Lan V.\ Truong, Matthew Aldridge, and Jonathan Scarlett
\thanks{L.~V.~Truong is with the Department of Engineering, University of Cambridge.  M.~Aldridge is with the School of Mathematics, University of Leeds.  J.~Scarlett is with the  Department of Computer Science, Department of Mathematics, and Institute of Data Science, National University of Singapore (NUS).  e-mail: \url{lt407@cam.ac.uk}; \url{m.aldridge@leeds.ac.uk}; \url{scarlett@comp.nus.edu.sg}

This work was supported by the Singapore National Research Foundation (NRF) under grant number R-252-000-A74-281.

Copyright \copyright~2017 IEEE. Personal use of this material is permitted.  However, permission to use this material for any other purposes must be obtained from the IEEE by sending a request to pubs-permissions@ieee.org
}
}

\maketitle
\begin{abstract} 
    In this paper, we study the problem of non-adaptive group testing, in which one seeks to identify which items are defective given a set of suitably-designed tests whose outcomes indicate whether or not at least one defective item was included in the test.  The most widespread recovery criterion seeks to exactly recover the entire defective set, and relaxed criteria such as approximate recovery and list decoding have also been considered.  In this paper, we study the fundamental limits of group testing under the significantly relaxed {\em weak recovery} criterion, which only seeks to identify a small fraction (e.g., $0.01$) of the defective items.  Given the near-optimality of i.i.d.~Bernoulli testing for exact recovery in sufficiently sparse scaling regimes, it is natural to ask whether this design additionally succeeds with much fewer tests under weak recovery.  Our main negative result shows that this is not the case, and in fact, under i.i.d.~Bernoulli random testing in the sufficiently sparse regime, an {\em all-or-nothing} phenomenon occurs:  When the number of tests is slightly below a threshold, weak recovery is impossible, whereas when the number of tests is slightly above the same threshold, high-probability exact recovery is possible.  In establishing this result, we additionally prove similar negative results under Bernoulli designs for the weak detection problem (distinguishing between the group testing model vs.~completely random outcomes) and the problem of identifying a single item that is definitely defective.  On the positive side, we show that all three relaxed recovery criteria can be attained using considerably fewer tests under suitably-chosen non-Bernoulli designs.  Thus, our results collectively indicate that when too few tests are available, naively applying i.i.d.~Bernoulli testing can lead to catastrophic failure, whereas ``cutting one's losses'' and adopting a more carefully-chosen design can still succeed in attaining these less stringent criteria.

%
\end{abstract}   

\section{Introduction} \label{eq:intro}

The group testing problem has recently regained significant attention following new applications and connections with compressive sensing; see \cite{Ald19} for a recent survey.  Briefly, the idea of group testing is to identify a small subset of defective items within a larger subset of items, based on a number of tests whose binary outcomes indicate whether or not at least one defective item was included in the test.

The standard recovery goal in group testing is to exactly identify the entire defective set.  In combinatorial group testing \cite{Kau64,Du93} a single test design is required to succeed for all defective sets up to a certain size, whereas in probabilistic group testing \cite{Mal78,Ald19} only high-probability recovery is required with respect to a random defective set (and/or a random test design).  Various relaxed recovery criteria have also appeared, including list decoding recovery \cite{Deb05,Dya83,Ras90,Che09,Ind10,Ngo11,Sca17} and approximate recovery criteria that allow a small number of false positives and/or false negatives in the reconstruction \cite{Lee15a,Sca17,Sca17,Sca15b}.

In this paper, focusing on probabilistic group testing, our goal is to better understand the fundamental limits of what can be achieved in the group testing problem under {\em significantly weaker} recovery criteria.  In particular, instead of asking when it is possible to recover most of the defectives, we seek to understand when it is possible just to {\em recover a small fraction}.  For general non-adaptive designs, we show that this goal can be obtained with much fewer tests, and identify an exact threshold on the number required.  On the other hand, for the widely-adopted i.i.d.\footnote{We write i.i.d.~as a shorthand for {\em independent and identically distributed}.}~Bernoulli test matrix design, we identify scenarios under which an {\em all-or-nothing} phenomenon occurs:  When the number of tests is slightly above a certain threshold, high-probability exact recovery is possible, whereas slightly below the same threshold, essentially nothing can be learned from the tests.  Thus, while Bernoulli designs can be near-optimal under standard recovery criteria, they are also prone to complete failure when there are too few tests, and accordingly can be highly suboptimal under relaxed recovery criteria.  Along the way, we additionally provide analogous results for the problems of weak detection and identifying a definite defective, which are formally defined in Section \ref{sec:overview}.


\subsection{Problem Setup} \label{sec:setup}

We consider a population of $p$ items indexed as $\{1,\dotsc,p\}$, and we let $k$ denote the number of defective items.  The set of defective items is denoted by $S$, and is assumed to be uniform over the ${p \choose k}$ possibilities.

A group testing procedure performs a sequence of $n$ tests, with $X^{(i)} \in \{0,1\}^p$ indicating which item is in the $i$-th test.  The resulting outcomes are $Y^{(i)} = \bigvee_{j \in S} X_j^{(i)}$ for $i=1,\dotsc,n$, where $n$ is the number of tests. That is, the test outcome is $1$ if there is any defective item in the test, and $0$ otherwise.  The tests can be represented by a matrix $\bX \in \{0,1\}^{n \times p}$, whose $i$-th row is $X^{(i)}$.  Similarly, the outcomes can be represented by a vector $\bY \in \{0,1\}^n$, whose $i$-th entry is $Y^{(i)}$.

In general, group testing procedures may be {\em adaptive} (i.e., $X^{(i)}$ may be chosen as a function of the previous outcomes) or {\em non-adaptive} (i.e., all $X^{(i)}$ must be selected prior to observing any outcomes).  We focus on the non-adaptive setting, which is often preferable in practice due to permitting highly parallelized tests.  
In particular, except where stated otherwise, we consider the widely-adopted {\em (i.i.d.) Bernoulli random test design} \cite[Sec.~2.1]{Ald19}, in which every item is independently placed in each test with probability $\frac{\nu}{k}$ for some $\nu > 0$, and we choose $\nu$ to be such that 
\begin{equation}
    \bigg(1 - \frac{\nu}{k}\bigg)^k = \frac{1}{2}. \label{eq:choice_nu}
\end{equation}
This choice ensures that the probability of a positive test is exactly $\frac{1}{2}$, which maximizes the entropy of each test outcome.  More importantly, this choice of $\nu$ leads to a provably optimal number of tests in broad scaling regimes, as we survey in Section \ref{sec:related}.  A simple asymptotic analysis gives $\nu = (\log 2)(1+o(1))$ as $k \to \infty$, which behaves similarly to the choice $\nu = \log 2$, but the exact choice described by \eqref{eq:choice_nu} will be more convenient to work with.  

\subsection{Related Work} \label{sec:related}

There have recently been numerous developments on theory and algorithms for probabilistic group testing \cite{Cha11,Che11,Mal13,Ald14a,Sca15b,Sca17b,Coj19,Coj19a} (see \cite{Ald19} for a survey); here we focus only on those most relevant to the present paper.

The most relevant works to us are those attaining upper and/or lower bounds on the number of tests of the form $\big( k \log_2 \frac{p}{k} \big) (1+o(1))$.  The most straightforward way that this quantity arises is that with $p \choose k$ possible defective sets and $2^n$ possible sequences of outcomes, we require $n \ge \log_2 {p \choose k}$ for each defective set to produce different outcomes.  In the sublinear regime $k = o(p)$, this simplifies to $n \ge \big( k \log_2 \frac{p}{k} \big) (1+o(1))$.  Building on this intuition, Fano's inequality was used in \cite{Mal78,Cha14} to show that $n \ge (1-\delta) \log_2 {p \choose k}$ is required to attain an error probability of at most $\delta$, and a refined bound $n \ge \log_2\big( {p \choose k} (1-\delta) \big)$ was established in \cite{Bal13}.

More recently, various results showed that $\big( k \log_2 \frac{p}{k} \big) (1+o(1))$ tests are {\em sufficient} for certain recovery guarantees under broad scaling regimes on $k$ as a function of $p$.  In \cite{Sca15b}, high-probability exact recovery was shown to be possible under Bernoulli random testing when $k = O(p^{1/3})$ and $n = \big( k \log_2 \frac{p}{k} \big) (1+o(1))$, and in addition, this result was extended to all $k = o(p)$ when the exact recovery criterion is replaced by the following approximate recovery criterion (see also \cite{Sca17}): Output a set $\widehat{S} \subseteq \{1,\dotsc,p\}$ of cardinality $k$ such that
\begin{equation}
    \max\big\{ |\widehat{S} \setminus S|, |S \setminus \widehat{S}| \big\} \le \alpha k \label{eq:partial}
\end{equation}
for some $\alpha \in (0,1)$.  The above-mentioned result holds for arbitrarily small $\alpha > 0$, as long as it is bounded away from zero as $p \to \infty$.  

On the other hand, the lower bounds for approximate recovery in \cite{Sca15b,Sca17} only state that in order to attain \eqref{eq:partial} for fixed $\alpha \in (0,1)$, it is necessary that $n \ge (1-\alpha)\big(k \log_2\frac{p}{k}\big) (1+o(1))$.  This suggests that as $\alpha$ increases, the constant factor in the number of tests could be reduced.  We will show that the preceding lower bound can in fact be matched  with an upper bound using a suitably-chosen non-Bernoulli design, whereas Bernoulli designs can fail even for $\alpha$ close to one when $n$ is slightly smaller than $k \log_2\frac{p}{k}$.


While our discussion thus far focuses on Bernoulli designs, in the case of exact recovery, improved bounds have been shown for a different random test design based on {\em near-constant tests-per-item} \cite{Joh16,Coj19,Coj19a}, in particular permitting $n = \big( k \log_2 \frac{p}{k} \big) (1+o(1))$ for all $k = O(p^{0.409})$ (improving on $O(p^{1/3})$).  However, upon moving to approximate recovery with parameter $\alpha$, both designs attain the threshold $n = \big( k \log_2 \frac{p}{k} \big) (1+o(1))$ in the limit as $\alpha \to 0$ \cite{Sca17,Hah19}, suggesting that there is less to be gained via the near-constant tests-per-item design under relaxed recovery criteria.  Nevertheless, extending our results to this design may be an interesting direction for future work.

Our work is inspired by recent studies of the all-or-nothing behavior of sparse linear regression under i.i.d.~Gaussian measurements; see \cite{Gam17} for a study of the maximum-likelihood estimator, and \cite{Ree19,Ree19a} for general estimators.  While group testing can be viewed as a non-linear Boolean counterpart to sparse linear regression \cite{Gil08,Ati12}, and our work will adopt the same high-level approach as \cite{Ree19}, the details will be very different.  


\subsection{Overview of the Paper}  \label{sec:overview}

As hinted above, in this paper, our main goal is to investigate the question of when the following mild recovery requirement is possible:
\begin{itemize}
    \item {\bf (Weak recovery)} Can we find a set $\hatS$ of size $k$ such that $|S \cap \hatS| \ge \delta k$ for small $\delta > 0$ with some non-zero constant probability?
\end{itemize}
The study of this goal essentially asks whether we can learn even a small amount of information from the test outcomes.  As a result, any hardness result (lower bound on the number of tests) under this criterion serves as a much stronger claim compared to a hardness result for exact recovery.   

While our main focus is on weak recovery, we will additionally address the following two recovery goals that are also much milder than exact recovery:
\begin{itemize}
    \item {\bf (Weak detection)} Can we perform a hypothesis test on $(\bX,\bY)$ to distinguish between the above group testing model and the ``null model'' in which $\bY$ is independent of $\bX$?
    \item {\bf (Identify a definite defective)} Can we identify just a single defective item, i.e., output a single index $\calI \in \{1,\dotsc,p\}$ with certainty that $\calI \in S$?  (Here we also allow ``detected errors'', in which the decoder declares that it is uncertain.)
\end{itemize}
We show in Section \ref{sec:positive} that these goals can be achieved with very few tests using suitably-chosen non-Bernoulli test designs.  On the other hand, for Bernoulli designs (studied in Section \ref{sec:bernoulli}), we will see that in sufficiently sparse regimes, these criteria all require essentially the same number of tests as exact recovery.

The main reason that we consider weak detection is that is serves as a useful stepping stone to establishing our negative result for weak recovery under Bernoulli designs.  Identifying a definite defective is also not our central focus, but its negative result comes almost for free via our analysis.  


Before proceeding, we briefly pause to discuss our emphasis on Bernoulli designs, despite our results demonstrating that they can be inferior to alternative designs under the relaxed recovery criteria.  The justification for doing so is that Bernoulli designs (and other related unstructured random designs) are widespread and extensively studied in the literature \cite{Ald19}, and thus serve as a standard ``go-to'' design.  As a result, it is essential to not only identify the cases that they succeed, but also understand their limitations.

\section{Positive Results for Bernoulli and non-Bernoulli Designs} \label{sec:positive}

To set the stage for our negative results on Bernoulli designs, we start by providing several positive (i.e., achievability) results for weak recovery, weak detection, and identifying a definite defective, focusing primarily on non-Bernoulli designs. 

\subsection{Asymptotically Optimal Approximate Recovery} \label{sec:opt_approx}

As discussed in Section \ref{sec:related}, the lower bounds for approximate recovery in \cite{Sca15b,Sca17} state that in order to attain
\begin{equation}
    \max \big\{ | \hatS \setminus S| , | S \setminus \hatS| \big\}\leq \alpha k \label{eq:lb}
\end{equation}
for fixed $\alpha \in (0,1)$, it is necessary that $n \geq (1-\alpha)(k \log_2 \frac{p}{k}) (1 - o(1))$.  The following result shows that one can in fact attain a matching upper bound for general test designs.

\begin{theorem} \label{genapprox1}
    {\em (Positive Result for Approximate Recovery)}
    Consider the probabilistic group testing problem, and for fixed $\alpha \in (0,1)$, suppose that $n \geq (1 + \eta) (1 - \alpha) \log_2 \binom pk$ for some $\eta > 0$. Then, when $k \to \infty$ with $k = o(p)$ as $p \to \infty$,  there exists a non-adaptive test design and decoder that outputs an estimate $\hatS$ of $S$ satisfying the following as $p \to \infty$:
    \begin{equation}
        \bbP \left[ \max \big\{ | \hatS \setminus S| , | S \setminus \hatS|\big\} \leq \alpha k \right] \to 1 .
    \end{equation}
    Moreover, it can be assumed that $\hatS$ has cardinality exactly $k$.
\end{theorem}

By letting $\alpha$ be arbitrarily close to one, this result establishes the following corollary.

\begin{corollary}
    {\em (Positive Result for Weak Detection)}
    Under the preceding setup, when $k \to \infty$ with $k = o(p)$ as $p \to \infty$, as long as $k = \Omega\big( k\log \frac{p}{k} \big)$ (with any implied constant), there exists a non-adaptive test design and decoder that outputs an estimate $\hat{S}$ of cardinality $k$ satisfying $|\hat{S} \cap S| = \Omega(k)$ with probability approaching one.
\end{corollary}

To prove Theorem \ref{genapprox1}, we will use the previous best-known positive result on approximate recovery as a stepping stone.  This is stated in the following lemma, whose main statement comes from \cite{Sca15b,Sca17}, with the second part regarding approximately-known $k$ instead coming from \cite[App.~B]{Sca18}. 

\begin{lemma} \label{genapprox2}
    {\em (Existing Positive Results for Approximate Recovery \cite{Sca15b,Sca17,Sca18})}
    Consider the probabilistic group testing problem with Bernoulli
    random testing using the choice of $\nu$ in (1), and suppose that $n \geq (1 + \eta) \log_2 \binom pk$ for some $\eta > 0$. Then, for any fixed $\alpha \in (0,1)$, when $k = o(p)$ as $p \to \infty$, there exists a decoder that outputs an estimate $\hatS$ of $S$ satisfying the following as $p \to \infty$:
    \begin{equation}
        \bbP\left[ \max \big\{ | \hatS \setminus S| , | S \setminus \hatS|\big\} \leq \alpha k \right] \to 1.
    \end{equation}
    Furthermore, this result remains true even when the decoder does not know the exact value of $k$ but instead only knows some quantity $\bark$ satisfying $\bark = k(1+o(1))$, and the $\nu/k$ test-inclusion probability is replaced by $\nu / \bark$.
\end{lemma}

Theorem \ref{genapprox1} reduces the number of tests in Lemma \ref{genapprox2} by a multiplicative factor of $1-\alpha$, and provides an asymptotically optimal result (including the constant).  In addition, this result demonstrates that weak recovery is possible whenever $n = \Omega\big(k \log_2 \frac{p}{k}\big)$ with an arbitrarily small implied constant.  In contrast, we will show in Theorem \ref{thm:weak} below that Bernoulli testing requires the implied constant to be one, and hence, the two differ by an arbitrarily large constant factor.

\begin{proof}[Proof of Theorem \ref{genapprox1}]
    The idea of the proof is straightforward:  We ignore slightly less than a fraction $\alpha$ of the items, and use Bernoulli testing to achieve approximate recovery on the items that were not ignored.

    More formally, fix $\alpha' \in (0,\alpha)$, and consider discarding $\alpha' p$ items chosen uniformly at random, leaving $p' = (1 - \alpha')p$ items remaining.  By Hoeffding's inequality for the Hypergeometric distribution \cite{Hoe63} and the assumption $k \to \infty$, we have with probability $1 - o(1)$ that the number of remaining defectives $k'$ satisfies 
    \begin{equation}
            k' = (1-\alpha')k \cdot (1+o(1)). \label{eq:k'}
    \end{equation}
    We apply the second part of Lemma \ref{genapprox2} on this reduced problem, with $\bark = (1-\alpha')k$ and an approximate recovery parameter $\alpha''$ to be selected shortly.  While the number of defectives $k'$ in the reduced problem is random, we see from \eqref{eq:k'} that it is known up to a multiplicative factor of $1+o(1)$, as required in Lemma \ref{genapprox2}.  Then, the number of tests $n'$ required in Lemma \ref{genapprox2} satisfies the following:
    \begin{align}
        n' &= (1 + \eta') \log_2 \binom {p'}{k'} \\
            &= (1 + \eta') \log_2 \binom {(1-\alpha')p'}{(1-\alpha')k'(1+o(1))} \\
            &= \big(1 + \eta' + o(1)\big) (1-\alpha') \log_2 \binom {p}{k}\label{eq:n'}
    \end{align}
    for arbitrarily small $\eta' > 0$, where we used $\log{p \choose k} = \big(k \log \frac{p}{k}\big)(1+o(1))$ for $k = o(p)$.  In addition, in accordance with the lemma, the returned set $\hatS'$ of size $k'$ contains at least
    \begin{equation}
        (1 - \alpha'')k' = (1-\alpha'')(1-\alpha')k \cdot (1+o(1)) \label{eq:at_least}
    \end{equation}
    defective items, with probability approaching one.

    It suffices to let the final estimate $\hatS$ equal $\hatS'$; alternatively, if an estimate of size $k$ is sought, one can add $k - k'$ arbitrary items to $\hatS'$ to form $\hatS$.  In either case, taking $\alpha'$ arbitrarily close to $\alpha$ and $\alpha''$ arbitrarily close to $0$, \eqref{eq:at_least} ensures that $\hatS$ contains at least $(1-\alpha) k$ defective items, which implies $\max \big\{ | \hatS \setminus S| , | S \setminus \hatS| \big\}\leq \alpha k$ due to the fact that $|S| = k$.  In addition, since $\eta' > 0$ is arbitrarily small, the number of tests in \eqref{eq:n'} simplifies to $(1+\eta)(1-\alpha) \log_2 {p \choose k}$ for arbitrarily small $\eta > 0$.

\end{proof}

We note that while the preceding analysis still uses Bernoulli testing as a subroutine via Lemma \ref{genapprox2}, the full $n \times p$ test matrix is not i.i.d.~Bernoulli, as a fraction $\alpha'$ of its columns are set to zero.  Hence, we still consider this to be a non-Bernoulli test design.

{\bf Discussion.}
Since the proof of Theorem \ref{genapprox1} is based on ignoring a fraction of the items, it amounts to a technique that immediately gives up on exact recovery, or ``cuts its losses'' from an early stage.  This raises the interesting question as to whether such an approach is actually necessary to obtain the bound in Theorem \ref{genapprox1}. 

To appreciate this distinction, note that Hwang's {\em adaptive} generalized binary splitting algorithm \cite{Hwa72} works by repeatedly identifying a single defective using at most $\big( \log_2 \frac{p}{k}\big)(1+o(1))$ tests, and then removing it from further consideration.  Hence, exact recovery is guaranteed with $n=\big( k \log_2 \frac{p}{k}\big)(1+o(1))$, but even if the procedure is stopped after $(1-\alpha) \big( k \log_2 \frac{p}{k}\big)(1+o(1))$ tests, one will have already identified $(1-\alpha) k$ defective items.  In this sense, Hwang's adaptive algorithm is {\em universally optimal} with respect to the approximate recovery parameter $\alpha$,\footnote{Note that the lower bound stated following \eqref{eq:lb} also holds for adaptive algorithms.} and the algorithm degrades gracefully as the number of tests decreases below the exact recovery threshold.

In contrast, the non-adaptive designs that we have considered do not enjoy such universality.  Under a Bernoulli design with $k = o(p)$, we can achieve approximate recovery with arbitrarily small $\alpha$ when $n = \big( k \log_2 \frac{p}{k} \big)(1+o(1))$ (Lemma \ref{genapprox2}), or even exact recovery if $k = O(p^{1/3})$ \cite{Sca15b}, but we are prone to complete failure for smaller $n$, at least in sufficiently sparse regimes (see Theorem \ref{thm:weak} below).  Alternatively, ignoring roughly a fraction $\alpha$ of the items leads to $\alpha$-approximate recovery when $n = (1-\alpha) \big( k \log_2 \frac{p}{k} \big)(1+o(1))$, but one retains this guarantee {\em and no better} regardless of how much $n$ is increased.  
Hence, it remains an interesting question for future work as to whether there exists a gracefully degrading (and ideally universally optimal) test matrix design in the non-adaptive setting.

\subsection{A Trivial Strategy for Weak Detection} \label{sec:weak_trivial}

We first describe the weak detection problem in more detail.  The goal is to distinguish between two joint distributions on the pair $(\bX,\bY)$ for some specified distribution $P(\bX)$:
\begin{itemize}
    \item Under distribution $P$, the X-marginal is $P(\bX)$, and the joint distribution $P(\bX,\bY)$ is deduced by deterministically computing $\bY$ from $\bX$ via the group testing model.
    \item Under distribution $Q$, the $\bX$ and $\bY$ marginals match those of $P$, but $\bX$ and $\bY$ are independent, i.e., $Q(\bX,\bY) = P(\bX)P(\bY)$.
\end{itemize}
This is a binary hypothesis testing problem.  The distribution $Q$ corresponds to ``completely uninformative outcomes'', so intuitively, if we cannot reliably distinguish between $P$ and $Q$, then we can view the group tests (under the distribution $P$) as being highly uninformative.

As hinted above, for Bernoulli designs, the main reason that we consider weak detection is as a stepping stone to providing a negative result for weak recovery (see Section \ref{sec:bernoulli}).  It turns out that if we consider weak detection as a standalone problem with general designs, then a trivial strategy succeeds with very few tests.  Formally, we have the following.

\begin{proposition} \label{prop:pos_weak_det}
    {\em (Trivial Strategy for Weak Detection)}
    For any number of tests $n$, there exists some distribution $P(\bX)$ such that given $(\bX,\bY)$, the group testing joint distribution $P(\bX,\bY)$ can be distinguished from $Q(\bX,\bY) := P(\bX)P(\bY)$ with zero error probability under $P$ and $1-2^{-n}$ error probability under $Q$.
\end{proposition}
\begin{proof}
    Consider letting each test independently contain all items with probability $\frac{1}{2}$ and no items with probability  $\frac{1}{2}$.  Hence, each outcome is independent and equiprobable on $\{0,1\}$.  Consider a detection algorithm that declares $P$ whenever the ones in $\bY$ exactly match the rows of ones in $\bX$, and declares $Q$ otherwise.  Then, under $P$, success is guaranteed due to the group testing model being noiseless.  Under $Q$, since the test outcomes are uniformly random, the probability of producing the correct outcomes is $2^{-n}$, proving the proposition.
\end{proof}

This result suggests that weak detection is of limited interest as a standalone recovery criterion in general.  Nevertheless, this finding arguably strengthens our subsequent negative results for Bernoulli designs (Section \ref{sec:bernoulli}), in the sense of showing that weak detection is not attained despite being an ``extremely easy'' goal.


\subsection{Weak Detection with Bernoulli Designs} \label{sec:weak_pos}

Our negative results for Bernoulli testing in Section \ref{sec:bernoulli} will demonstrate failure in sufficiently sparse regimes (e.g., $k = {\rm  poly}(\log p)$) when the number of tests is slightly below $\log_2{p \choose k}$.   On the other hand, fairly simple detection strategy can be used to attain the following positive result under Bernoulli testing when $k \gg p^{1/3}$ and $k = o(p)$.

\begin{theorem} \label{verymain3} 
    {\em (Positive Result for Weak Detection with Bernoulli Designs)}
    Consider the probabilistic group testing problem with Bernoulli random testing using the choice of $\nu$ in \eqref{eq:choice_nu}, and suppose that $n = (1-\eta)\log_2{p\choose k}$ for some $\eta \in (0,1)$. Then, under the condition that $k = o(p)$ and $k=\Omega\big(p^{\frac{1+\eta}{3+\eta}+\delta}\big)$ for some (arbitrarily small) fixed constant $\delta \in (0,1)$, there exists a binary hypothesis testing scheme  that succeeds in distinguishing $P(\bX,\bY)$ from $Q(\bX,\bY) := P(\bX)P(\bY)$ with probability approaching one as $p \to \infty$.
\end{theorem}
\begin{proof}
    The idea of is to perform weak detection based on the number of columns of $\bX$ that are {\em covered} by the test outcomes (i.e., the test outcome is one whenever the column entry is one).  The number of such columns is characterized by a binomial distribution under both the group testing model and the null model, and the distributions are shown to be statistically distinguishable under the conditions given.  The details are given in Appendix \ref{subB}.
\end{proof}

\subsection{A Simple Strategy for Identifying a Definite Defective} \label{sec:pos_dd}

We first describe the problem of {\em identifying a definite defective} in more detail.  We suppose that the decoder either outputs a single index $\calI \in \{1,\dotsc,p\}$ believed to be defective, or declares ``I don't know'' by outputting $\calI= 0$.  In the former scenario, we insist that $\calI$ must be defective (i.e., $\calI \in S$) with probability one, meaning that the only errors allowed are {\em detected errors} corresponding to $\calI= 0$.  This setup is partly motivated by the definite defectives algorithm for recovering the defective set \cite{Ald14a,Joh16}, as well as the notion of {\em zero undetected error capacity} in information theory \cite{For68}.

The following result shows that under a suitably-chosen non-Bernoulli test design, a single definitely-defective item can be reliably identified using only $O(\log\frac{p}{k})$ tests, representing a reduction by a factor of $k$ compared to the usual $O(k \log\frac{p}{k})$ scaling.

\begin{theorem} \label{thm:pos_dd}
    {\em (Positive Result for Identifying a Definite Defective)}
    Consider the probabilistic group testing problem, and suppose that $n \geq (1 + \eta)2c\log_2 \frac{p}{k}$ for some $\eta > 0$ and some positive integer $c > 0$. Then there exists a non-adaptive test design and decoder that outputs an estimate $\calI \in \{0,1,\dots,p\}$ of a definite defective (with $0$ representing ``I don't know'') satisfying
    \begin{equation}    
        \bbP[\calI > 0 \,\cap\, \calI \not\in S] = 0 ~~ \text{\rm and} ~~ \bbP[\calI = 0] \leq (1 - e^{-1} + o(1))^c
    \end{equation}
    as $k \to \infty$ and $p \to \infty$.
\end{theorem}
\begin{proof}
    We first consider the case $c = 1$.  Let $\calA$ be a uniformly random set of $\frac{p}{k}$ items.\footnote{We assume for simplicity that $\frac{p}{k}$ is an integer; the general case follows similarly by rounding.}  By the assumption that $k \to \infty$ and $k = o(p)$, it is straightforward to show that 
    \begin{equation}
        \bbP[|\calA \cap S| = 1] = e^{-1} + o(1). \label{eq:one_def}
    \end{equation}
    Indeed, the analogous claim is standard when each item is included in $\calA$ with probability $\frac{1}{k}$ \cite[Sec.~2.3]{Ald19}, and \eqref{eq:one_def} can then by understood by approximating the Hypergeometric distribution by the binomial distribution \cite{Soo96}.

    We proceed by describing a procedure from the SAFFRON algorithm of \cite{Lee15a} that is guaranteed to identify the single defective item in $\calA$ whenever $|\calA \cap S| = 1$, while also being able to identify with certainty whether $|\calA \cap S| = 1$ or $|\calA \cap S| \ne 1$.  This procedure uses $2v$ tests, where $v = \big\lceil \log_2 \frac{p}{k} \big\rceil = (1 + o(1))\log_2 \frac{p}{k}$.
        
    Number the items in $\calA$ from $0$ to $\frac{p}{k}-1$ in a fixed manner (e.g., maintaining the order that they take as items in $\{1,\dotsc,n\}$). For $i \in \big\{0, 1, \dots, \frac{p}{k}-1 \big\}$, let $\bb_i \in \{0,1\}^{2v}$ be a binary vector of length $2v$ and weight $v$ constructed as follows: The first $v$ entries are the number $i$ written in binary, and the last $v$ entries are the same, but with the $0$s and $1$s swapped. We then construct $2v$ tests, where test $j$ contains exactly the items corresponding to $i \in \calA$ for which the $j$-th entry of $\bb_i$ equals $1$.
    
    We now consider the following cases:
    \begin{itemize}
        \item If $\calA$ contains no defective items, then all $2v$ tests will be negative.  When this is observed, we set $\calI = 0$.
        \item If $\calA$ contains exactly one defective, then exactly $v$ of the tests will be positive.  When this is observed, we set $\calI$ to be item whose value $i \in \calA$ is spelled out (in binary) by the first $v$ test results.
        \item If $\calA$ contains two or more defectives, then more than $v$ of the tests will be positive.  When this is observed, we set $\calI = 0$.
    \end{itemize}
    The first and third cases ensure that we never erroneously set $\calI \ne 0$. In addition, we correctly identify a defective item in the second case, which occurs with probability $e^{-1} + o(1)$ due to \eqref{eq:one_def}.  This proves the theorem for $c = 1$.

    To handle $c > 1$, we simply repeat the preceding process $c$ times, drawing $\calA$ independently each time.  By doing so, we only fail if none of the sets $\calA$ contain exactly one defective item, which occurs with probability $(1-e^{-1} + o(1))^c$.
\end{proof}


We briefly mention that the $O\big(\log\frac{p}{k}\big)$ scaling in Theorem \ref{thm:pos_dd} cannot be improved further.  To see this, suppose by contradiction that a definite defective could be found with constant (non-zero) probability using $o\big(\log\frac{p}{k}\big)$ tests.  By repeating this procedure with uniformly random shuffling of the items, ignoring any $\calI = 0$ outcomes and removing any defectives identified, an adaptive group testing algorithm could recover the full defective set using $o\big(k \log\frac{p}{k}\big)$ tests on average.  However, this would violate the standard $\Omega\big(k \log\frac{p}{k}\big)$ lower bound \cite[Sec~1.4]{Ald19}, which holds even in the adaptive setting.

\section{All-or-Nothing Behavior for Bernoulli Designs} \label{sec:bernoulli}

To establish the an all-or-nothing threshold for weak recovery under Bernoulli designs, it is useful to first study the weak detection problem as a stepping stone.

%
%
\subsection{Weak Detection} \label{sec:distinguish}

Recall from Section \ref{sec:weak_pos} that the goal of weak detection is to distinguish $P(\bX,\bY)$ (defined according to the group testing model) from $Q(\bX,\bY) := P(\bX)P(\bY)$, with $P(\bX)$ being the i.i.d.~Bernoulli distribution using the choice of $\nu$ in \eqref{eq:choice_nu}, and $P(\bY) = 2^{-n}$.

For concreteness, we consider the Bayesian setting, where the observed pair $(\bX,\bY)$ is drawn from $P$ or $Q$ with probability $\frac{1}{2}$ each.  The resulting error probability of the hypothesis test is denoted by $\rvP_{\rme}$.  Trivially, choosing the hypothesis via a random guess gives $\rvP_{\rme} = \frac{1}{2}$.  It is a standard result in binary hypothesis testing that if $d_{\rm{TV}}(P,Q) \to 0$ as $p \to \infty$, then one cannot do better than random guessing asymptotically, i.e., it is impossible do better than $\rvP_{\rme} = \frac{1}{2} + o(1)$ (e.g., see \cite[Sec.~2.3.1]{DuchiNotes}).  

Let $D(P\|Q)$ denote the KL divergence between $P$ and $Q$.  By Pinsker's inequality, $D(P\|Q) \to 0$ implies $d_{\rm{TV}}(P,Q) \to 0$, and in addition, $D(P\|Q) \le \chi^2(P\|Q)$ \cite{Sas16}, where we define the $\chi^2$ divergence
\begin{align}
    &\chi^2(P \| Q) = \bE_Q\bigg[ \bigg( \frac{P(\bX,\bY)}{Q(\bX,\bY)} - 1 \bigg)^2 \bigg] \nonumber \\ &\quad = \bE_Q\bigg[ \bigg( \frac{P(\bX,\bY)}{Q(\bX,\bY)} \bigg)^2 \bigg] - 1 = \bE_P\bigg[ \frac{P(\bX,\bY)}{Q(\bX,\bY)} \bigg] - 1. \label{eq:chi_sq_forms}
\end{align}
Hence, to prove a hardness result for distinguishing $P$ from $Q$, it suffices to show that $\chi^2(P \| Q) \to 0$ as $p \to \infty$.  The following theorem gives conditions under which this is the case.

\begin{theorem}\label{verymain1} 
    {\em (Negative Result for Weak Detection)}
    Consider the probabilistic group testing problem with Bernoulli random testing using the choice of $\nu$ in \eqref{eq:choice_nu}, and suppose that $n \le (1-\eta)\log_2{p\choose k}$ for some $\eta \in (0,1)$. Then, when $k=o(p^{\frac{\eta}{1+\eta}})$ as $p \to \infty$, we have
    \begin{align}
        \label{tochi2}
        \chi^2(P\|Q)=\frac{1}{{p\choose k}}\sum_{l=0}^{k}{k \choose l}{p-k \choose l} 2^{n(1-\frac{l}{k})}-1\to 0
    \end{align} 
    as $p \to \infty$.  Hence, the smallest possible error probability for the binary hypothesis test between $P(\bX,\bY)$ and $Q(\bX,\bY) = P(\bX)(\bY)$ behaves as $\rvP_{\rme} = \frac{1}{2} + o(1)$.
\end{theorem}
\begin{proof}
    The proof is based on substituting the distributions $P$ and $Q$ into \eqref{eq:chi_sq_forms} and performing asymptotic simplifications.  The details are given in Appendix \ref{app:pf_neg}.
\end{proof}

The condition $k=o(p^{\frac{\eta}{1+\eta}})$ holds when $k$ grows sufficiently slowly with respect to $p$, e.g., it holds for arbitrarily small $\eta$ when $k = {\rm poly}(\log p)$.  On the other hand, it remains open as to whether a similar hardness result can be proved when $k$ grows faster than $\Theta(p^{\frac{\eta}{1+\eta}})$.  To address this question, we note the following:
\begin{itemize}
    \item Theorem \ref{verymain3} above shows that $P$ and $Q$ can be reliably distinguished when $n \ge (1-\eta)\log_2{p\choose k}$ and $k = p^{\frac{1+\eta}{3+\eta}+\Omega(1)}$, which essentially reduces to $k \gg p^{1/3}$ for small $\eta$.  Thus, $\log_2{p\choose k}$ no longer serves as an all-or-nothing threshold in this denser regime.
    This provides an interesting point of contrast with the analogous sparse linear regression problem \cite{Ree19}, where the analogous hardness result to Theorem \ref{verymain1} holds for all $k = O( \sqrt{p} )$.  In addition, \cite[App.~C]{Ree19} provides a positive result showing that this threshold is tight.
    \item Theorem \ref{verymain2} below shows that the above $\chi^2$-divergence does not approach zero when $n \ge (1-\eta)\log_2{p\choose k}$ and $k=\Omega\big(p^{\frac{\eta}{1+\eta}}\big)$.  Note that $\chi^2$-divergence approaching zero is sufficient, but not necessary, for establishing the hardness of distinguishing $P$ from $Q$.  Hence, this result does not establish such hardness, but it does show that any proof establishing hardness must move beyond the approach of bounding $\chi^2(P\|Q)$.

    In the sparse linear regression problem, a similar limitation regarding the $\chi^2$ divergence is overcome by conditioning out certain ``catastrophic'' low-probability events \cite{Ree19} that blow up the divergence.  Unfortunately, it appears to be difficult to identify an analogous event in the group testing problem.
\end{itemize}
Formally, the second of these is stated as follows.

\begin{theorem} \label{verymain2} 
    {\em (A Condition for Non-Vanishing $\chi^2$-Divergence)}
    Consider the probabilistic group testing problem with Bernoulli random testing using the choice of $\nu$ in \eqref{eq:choice_nu}, and suppose that $n = (1-\eta)\log_2{p\choose k}$ for some $\eta \in (0,1)$. Then, when $k=\Omega(p^{\frac{\eta}{1+\eta}})$ as $p \to \infty$, we have
    \begin{align}
        \label{fact2tobe}
        \liminf_{p \to \infty} \chi^2(P\|Q) \geq c
    \end{align} 
    for some constant $c>0$.
\end{theorem}
\begin{proof}
    See Appendix \ref{app:pf_neg}.
\end{proof}

%
%
\subsection{Weak Recovery} \label{sec:weak}

We are now ready to present our main negative result concerning the weak recovery criterion under Bernoulli testing, and establishing an all-or-nothing  threshold at $n \sim k\log_2\frac{p}{k}$.

\begin{theorem} \label{thm:weak}
    {\em (Negative Result for Weak Recovery)} Consider the probabilistic group testing problem with Bernoulli random testing using the choice of $\nu$ in \eqref{eq:choice_nu}, and suppose that $n \le (1-\eta)\log_2{p\choose k}$ for some $\eta \in (0,1)$. Then, when $k \to \infty$ with $k=o(p^{\frac{\eta}{1+\eta}})$ as $p \to \infty$, for any fixed $\alpha \in (0,1)$, any decoder that outputs some estimate $\widehat{S}$ of $S$ must yield the following as $p \to \infty$:
    \begin{equation}
        \bbP\Big[ \max\big\{ |\widehat{S} \setminus S|, |S \setminus \widehat{S}| \big\} \le \alpha k \Big] \to 0 \label{eq:no_weak}
    \end{equation}
\end{theorem}
\begin{proof}
    The proof is outlined as follows.  Following an idea used for sparse linear regression in \cite{Ree19}, we study the mutual information $I(S; \bY,Y' | \bX,X')$, where $(X',Y')$ is an additional test independent from $(\bX,\bY)$.  Combining some manipulations of information terms with the weak detection result of Theorem \ref{verymain1}, we show that $H(Y'|\bX,X',\bY) \to \log 2$ when $n \le (1-\eta)\log_2{p\choose k}$.  On the other hand, we show that if weak recovery were possible, we would be able to predict $Y'$ given $(\bX,X',\bY)$ better than random guessing, meaning that $H(Y'|\bX,X',\bY)$ would be bounded away from $\log 2$.  Combining these observations, we deduce that weak recovery must be impossible.  The details are given in Appendix \ref{app:pf_neg}.
\end{proof}


Hence, when $k$ is sufficiently sparse so that $k=o(p^{\frac{\eta}{1+\eta}})$ holds for any $\eta > 0$ (e.g., $k = {\rm poly}(\log p)$), the threshold $n^* = k\log_2\frac{p}{k}$ serves as an exact threshold between {\em complete success} and {\em complete failure}.  We note that in previous works, phase transitions were proved in \cite{Sca15b,Coj19,Coj19a} regarding the error probability of {\em recovering the exact defective set}, whereas Theorem \ref{thm:weak} gives the much stronger statement that one cannot even identify a small fraction of the defective set.

%
%
\subsection{Identifying a Definite Defective} \label{sec:identify_dd}

The following negative result for identifying a definite defective (see Section \ref{sec:pos_dd}) follows in a straightforward manner from our negative result for weak detection (Theorem \ref{verymain1}).


\begin{theorem} \label{thm:identify_dd}
    {\em (Negative Result for Identifying a Definite Defective)}
    Consider the probabilistic group testing problem with Bernoulli random testing using the choice of $\nu$ in \eqref{eq:choice_nu}, and suppose that $n \le (1-\eta)\log_2{p\choose k}$ for some $\eta \in (0,1)$. Then, when $k=o(p^{\frac{\eta}{1+\eta}})$ as $p \to \infty$, for any decoder that outputs an estimate $\calI \in \{0,1,\dotsc,p\}$ of a definite defective (with 0 representing ``I don't know''), we have
    \begin{equation}
        \bbP[\calI > 0 \cap \calI \notin S] = 0 \implies \bbP[\calI = 0] \to 1
    \end{equation}
    as $p \to \infty$.
\end{theorem}
\begin{proof}
    The idea is to note that since $\bbP[\calI > 0 \cap \calI \notin S] = 0$, the decoder must output $\calI = 0$ whenever there exists some $S'$ of cardinality $k$ that is {\em disjoint from $S$} and consistent with the test outcomes.  Using de Caen's bound \cite{Dec97}, we show that this holds with probability at least $\frac{1}{1+\chi^2(P\|Q)}$, and the theorem follows since $\chi^2(P\|Q) \to 0$ due to Theorem \ref{verymain1}.  The details are given in Appendix \ref{app:pf_neg}.
\end{proof}

\subsection{Discussion: Bernoulli vs.~General Designs} \label{sec:discussion}

We wrap up this section and Section \ref{sec:positive} by highlighting that for all three of the recovery criteria considered, Bernoulli designs can be highly suboptimal compared to general designs:
\begin{itemize}
    \item Weak recovery is possible (for general test designs) when $n = \Omega\big(k \log \frac{n}{k}\big)$ with an arbitrarily small implied constant, whereas Bernoulli testing requires $n = (1+o(1))k \log \frac{n}{k}$ in sufficiently sparse regimes;
    \item A trivial strategy exists for weak detection with only $O(1)$ tests (for constant error probability), whereas Bernoulli testing requires $n = (1+o(1))k \log \frac{n}{k}$ in sufficiently sparse regimes;
    \item Identifying a definite defective is possible with $O(\log n)$ tests, whereas Bernoulli testing requires $n = (1+o(1))k \log \frac{n}{k}$ in sufficiently sparse regimes.
\end{itemize}
Since Bernoulli designs are a prototypical example of an unstructured random design, these results indicate that despite their strong theoretical guarantees (e.g., as established in \cite{Sca15b}), significant care should be taken in adopting them when too few tests are available, or when relaxed recovery criteria are considered to be acceptable.  Essentially, in such cases, it is much better to ``cut one's losses'' early on via a hand-crafted design that targets the less stringent recovery criterion under consideration.

%
%
\section{Conclusion} \label{sec:conclusion}

We have established the fundamental limits of noiseless non-adaptive group testing under the weak detection criterion, and more generally approximate recovery, and shown Bernoulli designs to be highly suboptimal (in sufficiently sparse regimes) in the sense of exhibiting an all-or-nothing threshold at $n \sim k\log_2\frac{p}{k}$.  In addition, we gave similar negative results for Bernoulli testing under the criteria of weak detection and identifying a definite defective, while establishing that non-Bernoulli designs can attain these goals with very few tests.


Possible directions for future work include (i) closing the remaining gaps in between the positive and negative results in the regime $k = \Theta(p^{\theta})$ with $\theta \in \big(0,\frac{1}{3}\big)$; (ii) handling Bernoulli$\big(\frac{\nu}{k}\big)$ designs with more general choices of $\nu$; (iii) providing analogous results for the near-constant tests-per-item design \cite{Joh16,Coj19}; (iv) finding a design that attains exact or near-exact recovery with $n = (1+o(1))\big(k \log_2\frac{p}{k}\big)$ tests while degrading gracefully below this threshold; and (v) developing analogous results under random noise models \cite{Mal80,Cha11,Sca17b}.

\appendices

\section{Proof of Theorem \ref{verymain3} (Positive Result for Weak Detection with Bernoulli Designs)}\label{subB}


    By the assumption $k = o(p)$, we have $\log_2{p\choose k} = \big(k \log_2\frac{p}{k}\big)(1+o(1))$.  Hence, it suffices to prove the theorem for $n = (1-\eta)k \log_2\frac{p}{k}$, since we can incorporate the $1+o(1)$ term into $\delta$. 
    
    In the following, we use the terminology that the $j$-th column $\bX_j$ of $\bX$ is {\em covered} by $\bY$ if the support of $\bX_j$ is a subset of the support of $\bY$ (i.e., whenever the $i$-th entry of $\bX_j$ is $1$, the outcome $y_i$ is also $1$).  We consider distinguishing models $P$ and $Q$ by counting the number of columns of $\bX$ that are covered by $\bY$.
    
    Fix a constant $\zeta \in (0,1)$, and consider the ``typical'' set $\calT$ containing all the sequences $\by \in \{0,1\}^n$ such that the number of positive tests is between $\frac{n(1-\zeta)}{2}$ and $\frac{(1+\zeta )n}{2}$.  By the law of large numbers, we have $\bbP[\bY \notin \calT] \to 0$ as $p \to \infty$. 
    Given any sequence $\by \in \calT$, when an independent random column $\bX_j$ is generated, the (conditional) probability $q_0$ of it being covered satisfies
    \begin{align}
    \label{eq32}
        \bigg( 1 - \frac{\nu}{k} \bigg)^{\frac{(1+\zeta) n}{2}} \leq q_0 \leq   \bigg( 1 - \frac{\nu}{k} \bigg)^{\frac{(1-\zeta) n}{2}}. 
    \end{align}
    For $n = (1-\eta)k \log_2\frac{p}{k}$, recalling the choice of $\nu$ in \eqref{eq:choice_nu} (which  yields $P(\bY) = 2^{-n}$), we have
    \begin{align}
    	     \bigg( 1 - \frac{\nu}{k} \bigg)^{\frac{(1-\zeta) n}{2}} &= \bigg( \frac{1}{2} \bigg)^{\frac{(1-\zeta) n}{2k}} = \bigg(\frac{k}{p}\bigg)^{\frac{(1-\zeta)(1-\eta)}{2}},
    \end{align}
    and similarly
    \begin{align}
         \bigg( 1 - \frac{\nu}{k} \bigg)^{\frac{(1+\zeta) n}{2}}= \bigg(\frac{k}{p}\bigg)^{\frac{(1+\zeta)(1-\eta)}{2}}.
    \end{align}
    Hence, we have
    \begin{align}
        \label{boundp0}
        \bigg(\frac{k}{p}\bigg)^{\frac{(1+\zeta)(1-\eta)}{2}} \leq q_0 \leq \bigg(\frac{k}{p}\bigg)^{\frac{(1-\zeta)(1-\eta)}{2}}.
    \end{align}
    Then, the distribution of the number $\tilN(\bX,\by)$ of covered columns under the two hypotheses is given as follows:
    \begin{itemize}
        \item Under $P$: $\tilN(\bX,\by) \sim \big(k + \mbox{Binomial}(p - k, q_0)\big)$, where the addition of $k$ is due to the fact that the defective items' columns are almost surely covered due to the definition of the group testing model.
        \item Under $Q$: $\tilN(\bX,\by) \sim \mbox{Binomial}(p, q_0)$.
    \end{itemize}
    We consider the following procedure for distinguishing these two hypotheses:
    \begin{align}
        \tilN(\bX,\bY)\quad \mathop{\gtreqless}_{Q}^{P}\quad pq_0+\frac{k}{2}.
    \end{align}
    Then, given $\by$, the error probability $\rvP_{\rme}(\by)$ with a uniform prior satisfies 
    \begin{align}
        \rvP_{\rme}(\by) &= \frac{1}{2} \bbP_Q\bigg[\tilN(\bX,\by)> pq_0+\frac{k}{2}\bigg] \nonumber \\ &\quad + \frac{1}{2} \bbP_P\bigg[\tilN(\bX,\by)< pq_0+\frac{k}{2}\bigg]. \label{eq:hyp_pe}
    \end{align}
    For the first term in \eqref{eq:hyp_pe}, observe that by the Berry-Esseen Theorem~\cite{feller} (see Corollary~\ref{momlem} in Appendix \ref{sec:berry_esseen}), we have
    \begin{align} 
        &\bbP_Q\bigg[\tilN(\bX,\by)\geq  pq_0+\frac{k}{2}\bigg] \nonumber \\
        &~~ =\bbP_Q\bigg[\frac{\tilN(\bX,\bY) - p q_0}{\sqrt{p (1-q_0)q_0}} \geq  \frac{k}{2\sqrt{p (1-q_0)q_0}}\bigg]\\
        \label{eq45new}
        &~~ \leq \calQ\bigg( \frac{k}{2\sqrt{p (1-q_0)q_0}}\bigg) +\frac{6\rho}{\sigma^3 \sqrt{p}},
    \end{align}
    where $\calQ(t) = \frac{1}{\sqrt{2\pi}}\int_{t}^{\infty}e^{-u^2/2}\,{\rm d}u$ denotes the standard Gaussian upper tail probability function, and as also shown in Appendix \ref{sec:berry_esseen}, the relevant moments are
    \begin{align}
        \rho &= (1-q_0)^3 q_0 + q_0^3 (1-q_0),\\
        \sigma &=\sqrt{(1-q_0)q_0}.
    \end{align}
    The ratio appearing in \eqref{eq45new} can be simplified as follows:
    \begin{align}
        \frac{6\rho}{\sigma^3 \sqrt{p}}&=\frac{6(1-q_0)^3 q_0 + 6q_0^3 (1-q_0)}{(\sqrt{(1-q_0)q_0})^3 \sqrt{p}}\\
        &= \frac{6(1-q_0)^2 q_0 + 6q_0^3 }{q_0^{3/2}\sqrt{(1-q_0)p} }\\
        &\leq \frac{12 q_0}{q_0^{3/2}\sqrt{(1-q_0)p} }\\
        \label{eq49new}
        &=\frac{12}{\sqrt{p q_0(1-q_0)}}.
    \end{align}
    Similarly, again using the Berry-Essen theorem, and writing $\tilN$ in place of $\tilN(\bX,\bY)$ for brevity, we have
    \begin{align} 
        &\bbP_P\bigg[\tilN \leq  pq_0+\frac{k}{2}\bigg] \nonumber  \\
        &=\bbP_P\bigg[\tilN -k - (p-k) q_0 \leq pq_0+ \frac{k}{2}-k - (p-k) q_0 \bigg]\\
        &=\bbP_P\bigg[\frac{\tilN -k - (p-k) q_0}{\sqrt{(p-k) (1-q_0)q_0}} \leq \frac{q_0 k- \frac{k}{2}}{\sqrt{(p-k) (1-q_0)q_0}} \bigg]\\
        \label{eq50new}
        &\leq 1-\calQ\bigg( \frac{q_0 k-\frac{k}{2}}{\sqrt{(p-k) (1-q_0)q_0}}\bigg) +\frac{6\rho}{\sigma^3 \sqrt{p-k}},
    \end{align}
    and since $k = o(p)$, we have from \eqref{eq49new} that $\frac{6\rho}{\sigma^3 \sqrt{p-k}} \le \frac{12}{\sqrt{p q_0(1-q_0)}} (1+o(1))$.
    
    We know from \eqref{boundp0} that $q_0 \to 0$, and combining this with $k = o(p)$, we see from \eqref{eq49new} and~\eqref{eq50new} that $\rvP_{\rme}\to 0$ as long as $p q_0\to \infty$ and $k=\omega(\sqrt{pq_0})$. The condition $p q_0\to \infty$ follows as an immediate consequence of \eqref{boundp0} (with $k = o(p)$ and $\delta < 1$).  In addition, again using \eqref{boundp0}, we find that 
    the condition $k=\omega\big(\sqrt{p q_0}\big)$ is satisfied if
    \begin{align}
        k&=\omega\bigg(\sqrt{p} \bigg(\frac{k}{p}\bigg)^{\frac{(1-\eta)(1-\zeta)}{4}}\bigg). 
    \end{align}
    Letting $a = (1-\eta)(1-\zeta)$, we find that this condition holds if $k^{1 - \frac{a}{4}} = \omega\big( p^{\frac{1}{2} - \frac{a}{4}} \big)$, which simplifies to $k = \omega\big( p^{\frac{2-a}{4-a}} \big)$.  Substituting $a = (1-\eta)(1-\zeta)$, and recalling that $\eta$ is constant and $\zeta$ is arbitrarily small, we find that the preceding condition holds as long as
    \begin{align}
        \label{scar1}
  		k&=\Omega\big(p^{\frac{1+\eta}{3+\eta} +\delta}\big)
    \end{align} 
    for arbitrarily small $\delta \in (0,1)$.  This completes the proof of Theorem \ref{verymain3}.

%
%

\section{Proofs of Negative Results for Bernoulli Designs} \label{app:pf_neg}

\subsection{Preliminary Calculations} \label{sec:prelim}


Since the group testing model $P$ and the null model $Q$ have the same $\bX$ distribution, and the null model assigns probability $2^{-n}$ to each $\bY$ sequence (due to the choice of $\nu$ in \eqref{eq:choice_nu}), we have
\begin{equation}
\label{eq3}
    \frac{P(\bX,\bY)}{Q(\bX,\bY)} = 2^n P(\bY|\bX) = 2^n \sum_{S} \frac{1}{{p \choose k}} P(\bY|\bX,S), 
\end{equation}
where here and subsequently, the summation over $S$ is implicitly over all $p \choose k$ subsets of $\{1,\dotsc,n\}$ of cardinality $k$.  Since the observation model defining $P$ is deterministic, $P(\bY|\bX,S)$ is simply $1$ if $\bY$ is consistent with $(\bX,S)$ (according to $Y^{(i)} = \bigvee_{j \in S} X_j^{(i)}$), and $0$ otherwise.  Letting $\calI_{S}(\bX,\bY)$ be the corresponding indicator function, we rewrite \eqref{eq3} as
\begin{equation}
    \label{eq4}
    \frac{P(\bX,\bY)}{Q(\bX,\bY)} = \frac{2^n}{{p \choose k}} \sum_{S} \calI_{S}(\bX,\bY),
\end{equation}
and take the square to obtain
\begin{equation}
    \label{senti2}
    \bigg(\frac{P(\bX,\bY)}{Q(\bX,\bY)}\bigg)^2 = \frac{4^n}{{p \choose k}^2} \sum_{S,S'} \calI_{S}(\bX,\bY) \calI_{S'}(\bX,\bY).
\end{equation}
Taking the average over $(\bX,\bY) \sim Q$ and using the middle form in \eqref{eq:chi_sq_forms}, we obtain
\begin{equation}
    \label{eqbusy}
    \chi^2(P\|Q) = \frac{4^n}{{p \choose k}^2} \sum_{S,S'} \bbE_Q\big[ \calI_{S}(\bX,\bY) \calI_{S'}(\bX,\bY) \big] - 1.
\end{equation}
The average here is the probability that a randomly generated $(\bX,\bY)$ (independent of each other) is consistent with both $S$ and $S'$.  By the symmetry of $(\bX,\bY)$ with respect to re-labeling items, we can assume without loss of generality that $S$ equals the set
\begin{align}
    \label{defsetS0}
    S_0 = \{1,2,\dotsc,k\},
\end{align}
and average over $S'$ alone; by splitting into $S'$ with $\ell$ entries in $\{k+1,\dotsc,p\}$ (non-overlapping with $S_0$) and $k-\ell$ entries in $\{1,\dotsc,k\}$ (overlapping with $S_0$), we obtain
\begin{align}
    &\chi^2(P\|Q) \nonumber \\ &= \frac{4^n}{{p \choose k}} \sum_{\ell=0}^k {k \choose \ell}{p-k \choose \ell} \bbE_Q\big[ \calI_{S_0}(\bX,\bY) \calI_{S'}(\bX,\bY) \big]  - 1, \label{eq:chi_sq}
\end{align}
where $S'$ implicitly satisfies $|S_0 \cap S'| = k-\ell$ in the expectation.

Now observe that $\bbE_Q\big[ \calI_{S_0}(\bX,\bY) \calI_{S'}(\bX,\bY) \big]$ is the probability (with respect to $Q$) that every one of the $n$ tests satisfies {\em any one} of the following:
\begin{itemize}
    \item The test outcome is negative, and all $k + \ell$ items from $S_0 \cup S'$ are excluded;
    \item The test outcome is positive, and at least one item from $S_0 \cap S'$ is included;
    \item The test outcome is positive, and no items from $S_0 \cap S'$ are included, but at least one item from each of $S_0 \setminus S'$ and $S' \setminus S_0$ are included.
\end{itemize}
For a single test, we characterize the probabilities of these three events under $Q$ as follows follows (recalling \eqref{eq:choice_nu}):
\begin{itemize}
    \item The first event has probability $\frac{1}{2} \big( 1 - \frac{\nu}{k} \big)^{k + \ell} = \frac{1}{2} \cdot \big(\frac{1}{2}\big)^{\frac{k+\ell}{k}} = \big(\frac{1}{2}\big)^{2+\frac{\ell}{k}}$;
    \item The union of the second and third events above can be reformulated as the event the test outcome is positive and {\em none} of the following events occur: (i) All items from $S_0 \cup S'$ are excluded; (ii) All items from $S_0$ are excluded, but at least one from $S' \setminus S_0$ is included; (iii) All items from $S'$ are excluded, but at least one from $S_0 \setminus S'$ is included.  
    Using this formulation, the union of the second and third events above has probability $\frac{1}{2}\big[ 1 - \big(1 - \frac{\nu}{k}\big)^{k + \ell} - 2 \cdot \frac{1}{2} \cdot  \big(1 - \big(1 - \frac{\nu}{k}\big)^{\ell}\big) \big] = \frac{1}{2}\big[ 1 - \big(\frac{1}{2}\big)^{1 + \frac{\ell}{k}} - \big(1 - \big(\frac{1}{2}\big)^{\frac{\ell}{k}}\big) \big] = \big(\frac{1}{2}\big)^{1+\frac{\ell}{k}} - \big(\frac{1}{2}\big)^{2+\frac{\ell}{k}} = \big(\frac{1}{2}\big)^{2+\frac{\ell}{k}}. $
\end{itemize}
Summing these two probabilities together gives an overall probability of $2 \cdot \big(\frac{1}{2}\big)^{2+\frac{\ell}{k}} = \big(\frac{1}{2}\big)^{1+\frac{\ell}{k}}$ associated with a single test. 
Since the tests are independent, taking the intersection of the corresponding $n$ events gives
\begin{align}
    \label{keykey}
    \bbE_Q\big[ \calI_{S_0}(\bX,\bY) \calI_{S'}(\bX,\bY) \big]=2^{-n(1+\frac{\ell}{k})},
\end{align}
and substitution into \eqref{eq:chi_sq} yields
\begin{equation}
    \chi^2(P\|Q) = \frac{4^n}{{p \choose k}} \sum_{\ell=0}^k {k \choose \ell}{p-k \choose \ell} 2^{-n(1+\frac{\ell}{k})}  - 1. \label{eq:chi_sq_2}
\end{equation}

\subsection{Proof of Theorem \ref{verymain1} (Impossibility of Weak Detection for $n \le (1-\eta)\log_2{p\choose k}$)} \label{subA}

Recall the setup of weak detection described in Section \ref{sec:weak_trivial}, and that we consider i.i.d.~Bernoulli designs with $\nu$ chosen via \eqref{eq:choice_nu}.  We first prove the following lemma, which provides an upper bound on the $\chi^2$-divergence.

\begin{lemma} \label{lem_simple} 
    Assume that $n \leq (1-\eta)\log_2{p\choose k}$ for some $\eta \in (0,1)$. Then, 
    it holds that
    \begin{align}
    \label{fact1to}
    \chi^2(P\|Q) 
    \leq \exp\bigg[e^{1-\eta}k \bigg(\frac{k}{p}\bigg)^{\eta} \frac{p}{p-k+1}\bigg]-1.
    \end{align}
\end{lemma}
\begin{proof}
    Using the assumption $n\leq (1-\eta) \log_2 {p \choose k}$, we bound \eqref{eq:chi_sq_2} as follows:
    \begin{align}
        &\chi^2(P\|Q) \nonumber \\
        &~~=\frac{1}{{p\choose k}}\sum_{l=0}^{k}{k \choose l}{p-k \choose l} 2^{n(1-\frac{l}{k})}-1 \label{follow1} \\
        &~~\leq \frac{1}{{p\choose k}}\sum_{l=0}^{k}{k \choose l}{p-k \choose l} 2^{(1-\frac{l}{k}) (1-\eta) \log_2{p \choose k} }-1 \label{follow2}  \\
        &~~= \frac{1}{{p\choose k}}\sum_{l=0}^{k}{k \choose l}{p-k \choose l} {p \choose k}^{(1-\eta)(1-\frac{l}{k})}-1\\
        \label{google1}
        &~~={p \choose k}^{-\eta}\sum_{l=0}^{k}{k \choose l}{p-k \choose l} {p \choose k}^{-(1-\eta)\frac{l}{k}}-1. 
    \end{align}
    In addition, for all $l<k$, we have
    \begin{align}
        \frac{{p-k \choose l}}{{p \choose k}}&\leq \frac{{p \choose l}}{{p \choose k}}\\
        &=\frac{k!(p-k)!}{l!(p-l)!}\\
        &=\frac{k(k-1)\cdots (l+1)}{(p-l)(p-l-1)\cdots (p-k+1)}\\
        \label{google2}
        &\leq \bigg(\frac{k}{p-k+1}\bigg)^{k-l}.
    \end{align}
    Since~\eqref{google2} also holds for $l=k$, it follows that
    \begin{align}
        \label{google3}
        \frac{{p-k \choose l}}{{p \choose k}} \leq \bigg(\frac{k}{p-k+1}\bigg)^{k-l}
    \end{align} for all $0\leq l \leq k$.
    
    From~\eqref{google1} and~\eqref{google3}, we obtain 
    \begin{align}
        &\chi^2(P\|Q) \nonumber  \\ 
        &\leq {p \choose k}^{1-\eta}\sum_{l=0}^{k}{k \choose l}\bigg(\frac{k}{p-k+1}\bigg)^{k-l} {p \choose k}^{-(1-\eta)\frac{l}{k}}-1\\
        &={p \choose k}^{1-\eta}\bigg(\frac{k}{p-k+1}\bigg)^k \nonumber \\
            &\qquad\qquad \times \sum_{l=0}^{k}{k \choose l}\bigg(\frac{k}{p-k+1}\bigg)^{-l} {p \choose k}^{-(1-\eta)\frac{l}{k}}-1\\
        \label{note1}
				&={p \choose k}^{1-\eta}\bigg(\frac{k}{p-k+1}\bigg)^k\bigg[1+\frac{p-k+1}{k}{p \choose k}^{-\frac{1-\eta}{k}}\bigg]^k-1\\
        &=\bigg[{p \choose k}^{\frac{1-\eta}{k}}\bigg(\frac{k}{p-k+1}\bigg)+1\bigg]^k-1\\
        \label{qote}
        &\leq \bigg[\bigg(\frac{pe}{k}\bigg)^{1-\eta} \bigg(\frac{k}{p-k+1}\bigg)+1\bigg]^k-1\\
        &\leq \exp\bigg[k \bigg(\frac{pe}{k}\bigg)^{1-\eta} \bigg(\frac{k}{p-k+1}\bigg)\bigg]-1 \label{interact0}\\
        \label{interfact}
        &=\exp\bigg[e^{1-\eta}k \bigg(\frac{k}{p}\bigg)^{\eta} \frac{p}{p-k+1}\bigg]-1
    \end{align} 
    where \eqref{note1} follows from the fact that $\sum_{j=0}^k {k \choose j} x^j = (1+x)^k$, \eqref{qote} follows from~${p\choose k} \leq \big(\frac{p e}{k}\big)^k$, and \eqref{interact0} uses $1+x \le e^{x}$.  This proves \eqref{fact1to}.
\end{proof}

To prove Theorem \ref{verymain1}, it suffices to show that the right-hand side of \eqref{fact1to} tends to zero as $p \to \infty$.  To see this, observe that the condition $k=o(p^{\frac{\eta}{1+\eta}})$ can equivalently be written as $k \big(\frac{k}{p}\big)^{\eta}=o(1)$, and this condition implies that
\begin{align}
    0\leq \chi^2(P\|Q) \leq \exp\bigg[e^{1-\eta}k \bigg(\frac{k}{p}\bigg)^{\eta} \frac{p}{p-k+1}\bigg]-1 \to 0
\end{align} 
as $p \to \infty$. This means that $\chi^2(P\|Q)  \to 0$ when $k=o(p^{\frac{\eta}{1+\eta}})$, which proves Theorem \ref{verymain1}.

\subsection{Proof of Theorem \ref{verymain2} (A Condition for Non-Vanishing $\chi^2$-Divergence)}\label{prflem1add}

Here we prove that $\liminf_{p \to \infty} \chi^2(P\|Q) \geq c$ for some $c>0$ when $n\geq (1-\eta) \log_2{p \choose k}$ under the assumptions $k = o(p)$ and $k=\Omega(p^{\frac{\eta}{1+\eta}})$.\footnote{We are grateful to an anonymous reviewer for helping us to significantly simplify this proof.}

As mentioned in Section \ref{sec:distinguish}, $\chi^2(P\|Q) \to 0$ implies that $d_{\rm{TV}}(P,Q)\to 0$, which in turn implies that weak detection is impossible for any algorithm. However, Theorem \ref{verymain3} shows that weak detection is possible when $k=\Omega(p^{\frac{1+\eta}{3+\eta}+\delta})$ for arbitrarily small $\delta \in (0,1)$. Note that that $\frac{1+\eta}{3+\eta} < \frac{1}{2}$ for any $\eta \in [0,1)$. Thus, it must be the case that $\liminf_{p \to \infty} \chi^2(P\|Q) \ge c$ whenever $k=\Omega(\sqrt{p})$, since otherwise $\chi^2(P\|Q)=o(1)$ would be a contradiction.  In the following, we therefore only consider $k = o(\sqrt p)$.

We need to prove that $\liminf_{p \to \infty} \chi^2(P\|Q) \geq c$ for some $c>0$ when $n\geq (1-\eta) \log_2{p \choose k}$ under the assumptions $k = o(\sqrt{p})$ and $k=\Omega(p^{\frac{\eta}{1+\eta}})$.
Following \eqref{follow1}--\eqref{google1} with the     	inequality in \eqref{follow2} reversed, we have
	\begin{align}
	\chi^2(P\|Q)\geq {p \choose k}^{-\eta}\sum_{l=0}^{k}{k \choose l}{p-k \choose l} {p \choose k}^{-(1-\eta)\frac{l}{k}}-1. \label{chi_initadd}
	\end{align}
Further lower bounding the right-hand side by taking only the terms $l \in \{k-1,k\}$, we obtain
\begin{align}
\chi^2(P\|Q) &\geq -1+ \frac{{p-k\choose k}}{{p\choose k}}+k \frac{{p-k \choose k-1}}{{p\choose k}}{p\choose k}^{\frac{1-\eta}{k}}\\
&=-1+ \frac{{p-k\choose k}}{{p\choose k}}+ \frac{k^2}{p-2k+1}\frac{{p-k\choose k}}{{p\choose k}}{p\choose k}^{\frac{1-\eta}{k}} \label{cu0}\\
&=-1+\frac{{p-k\choose k}}{{p\choose k}}\bigg[1+\frac{k^2}{p-2k+1}{p\choose k}^{\frac{1-\eta}{k}}\bigg]\\
&\geq -1+\frac{{p-k\choose k}}{{p\choose k}}\bigg[1+\frac{k^2}{p-2k+1}\bigg(\frac{p}{k}\bigg)^{1-\eta}\bigg] \label{cu1}\\
&=-1+\frac{{p-k\choose k}}{{p\choose k}}\bigg[1+\frac{k^{1+\eta}}{p^{\eta}}\bigg(\frac{p}{p-2k+1}\bigg)\bigg] \label{cu2},
\end{align} 
where \eqref{cu0} uses ${ {p-k} \choose {k-1} } = { {p-k} \choose {k} }\frac{k}{p-2k+1}$, and \eqref{cu1} follows since ${p\choose k}\geq (\frac{p}{k})^k$.  To bound the ratio of binomial coefficients in \eqref{cu2}, observe that
\begin{align}
1&\geq\frac{{p-k\choose k}}{{p\choose k}}\\
&=\frac{(p-k)(p-k-1)\cdots (p-2k+1)}{p(p-1)\cdots(p-k+1)}\\
&> \bigg(1-\frac{2k}{p}\bigg)^k\\
&\geq 1-\frac{2k^2}{p} \label{eq:last}
\end{align} 
where \eqref{eq:last} follows since $(1+x)^k \geq 1+kx$ for all $x>-1$. Hence, since we are considering  $k=o(\sqrt{p})$, it follows that $\frac{{p-k\choose k}}{{p\choose k}}=1+o(1)$.  Combining this with \eqref{cu2} and the assumption $k=\Omega(p^{\frac{\eta}{1+\eta}})$, we obtain
\begin{align}
\liminf_{p \to \infty} \chi^2(P\|Q) \ge c
\end{align} 
for some constant $c>0$, completing the proof of Theorem \ref{verymain2}.

\subsection{Proof of Theorem \ref{thm:weak} (Negative Result for Weak Recovery)} \label{app:weak}

As mentioned in Section \ref{sec:distinguish}, $\chi^2(P\|Q) \to 0$ implies that $D(P\|Q) \to 0$.  Consider $(\bX,\bY) \sim P$, along with an additional pair $(X',Y') \in \{0,1\}^p \times \{0,1\}$ drawn from the same joint distribution as a single test in $(\bX,\bY)$, independently from $(\bX,\bY)$.  Following the steps of \cite{Ree19} for sparse linear regression, we consider the following conditional mutual information term:
\begin{align}
    &I(S; \bY,Y' | \bX,X') \nonumber \\
        &\quad = \bE_P\bigg[ \log\frac{P(\bY,Y'|\bX,X',S)}{P(\bY,Y'|\bX,X')} \bigg] \\
        &\quad = \bE_P\bigg[ \log\frac{P(\bY,Y'|\bX,X',S)}{Q(\bY,Y')} \bigg]  \nonumber \\
            & \qquad\qquad - \bE_P\bigg[ \log\frac{P(\bY,Y'|\bX,X')}{Q(\bY,Y')} \bigg] \\
        &\quad = \bE_P\bigg[ \log\frac{P(\bY,Y'|\bX,X',S)}{Q(\bY,Y')} \bigg] - D(P\|Q),
\end{align}
where $D(P\|Q)$ is now defined according to $P$ and $Q$ containing $n+1$ tests instead of $n$ (one extra for $X'$, $Y'$).  Under $P$, we have $P(\bY,Y'|\bX,X',S) = 1$ almost surely, and combining this with $Q(\bY,Y') = 2^{-(n+1)}$, it follows that
\begin{equation}
    I(S; \bY,Y' | \bX,X') = (n+1)\log 2 - D(P\|Q). \label{eq:mi1}
\end{equation}
Moreover, by the chain rule for mutual information, we have
\begin{align}
    I(S; \bY,Y' | \bX,X') 
        &= I(S; \bY | \bX, X') + I(S; Y' | \bX, X', \bY) \\
        &\le n \log 2 + H(Y' | \bX, X', \bY), \label{eq:mi2}
\end{align}
where the two terms are attained as follows by expanding the conditional mutual information as as a difference of conditional entropies:
\begin{itemize}
    \item For the first term, write $I(S; \bY | \bX, X') = H(\bY | \bX, X') - H(\bY | \bX, X', S) \le H(\bY | \bX, X')$, and note that $H(\bY | \bX, X') \le n \log 2$ since $\bY \in \{0,1\}^n$ and entropy is upper bounded by the logarithm of the number of outcomes;
    \item For the second term, write $I(S; Y' | \bX, X', \bY) = H(Y' | \bX, X', \bY) - H(Y' | \bX, X', \bY, S)$, and note that we have $H(Y' | \bX, X', \bY, S) = 0$ since $Y'$ is deterministic given $(X',S)$. 
\end{itemize}
Combining \eqref{eq:mi1} and \eqref{eq:mi2} gives
\begin{equation}
    H(Y' | \bX, X', \bY) \ge \log 2 - D(P\|Q) = \log 2 - o(1), \label{eq:H_bound}
\end{equation}
since $D(P\|Q) \to 0$ for $n \le (1-\eta) \log_2{p \choose k}$ by Theorem \ref{verymain1} (the replacement of of $n$ by $n+1$ only amounts to a negligible multiplicative $1+o(1)$ change in $\eta$).  
Since the entropy functional is continuous, and the entropy of a binary random variable is at most $\log 2$ with equality if and only if the random variable is equiprobable on its two values, we deduce from \eqref{eq:H_bound} that the following holds: With probability $1-o(1)$ with respect to $(\bX,X',\bY)$, the conditional distribution of $Y'$ places probability $\frac{1}{2} + o(1)$ on each of $Y'=0$ and $Y'=1$.

To complete the proof of Theorem \ref{thm:weak}, we show that the preceding claim precludes the possibility of weak recovery, i.e., \eqref{eq:no_weak} holds.  Suppose by contradiction to \eqref{eq:no_weak} that it were possible to use $(\bX,\bY)$ to attain $|S \cap \hatS| \ge \delta k$ with probability at least $\delta$, for some $\delta > 0$.  In the following, we assume the extreme case $|S \cap \hatS| = \delta k$; the case of strict inequality follows similarly.  Consider a procedure that uses this $\hatS$ to construct an estimator that takes the test vector $X'$ as input and returns an estimate $\hat{Y}'$ of $Y'$ as follows: Set $\hat{Y}' = 1$ if the test includes any item from $\hatS$, and $\hat{Y}' = 0$ otherwise.  There are two scenarios in which the estimate is incorrect:
\begin{itemize}
    \item The test may include no items from $S$ (and hence $Y' = 0$), but an item from $S' \setminus S$ (and hence $\hat{Y}' = 1$).  By the choice of $\nu$ in \eqref{eq:choice_nu}, the probability (with respect to $P$) of this occurring is $\frac{1}{2}\big(1 - \big(1-\frac{\nu}{k}\big)^{(1-\delta)k}\big) = \frac{1}{2}\big(1-\big(\frac{1}{2}\big)^{1-\delta}\big)$.
    \item The test may include no items from $S'$ (and hence $\hatY' = 0$), but an item from $S \setminus S'$ (and hence $Y' = 1$).  By the same argument as above, the probability of this occurring is $\frac{1}{2}\big(1-\big(\frac{1}{2}\big)^{1-\delta}\big)$.
\end{itemize} 
Hence, when $|S \cap \hatS| = \delta k$, the estimator produces $\hatY' = Y'$ with probability $\big(\frac{1}{2}\big)^{1-\delta}$.  As a result, for any fixed $\delta > 0$, the success probability behaves as $\frac{1}{2} + \Omega(1)$.  This is in contradiction with the conditional distribution of $Y'$ stated following \eqref{eq:H_bound} (which only permits a $\frac{1}{2} + o(1)$ probability of correctness), and this completes the proof by contradiction establishing \eqref{eq:no_weak}.

\subsection{Proof of Theorem \ref{thm:identify_dd} (Negative Result for Identifying a Definite Defective)} \label{app:identify_dd}

Consider any algorithm that, with probability one, only outputs $\calI \ne 0$ when $\calI$ is the index of a defective item.
    If $S$ is the true defective set, then it is easy to see that an error occurs (i.e., $\calI = 0$) if some $S'$ disjoint from $S$ is still consistent with $(\bX,\bY)$.  Denoting this event by $\calE_{S'}$, it follows that
    \begin{align}
        \bbP[\calI = 0 \,|\, S] 
            &\ge \bbP\bigg[ \bigcup_{S' \,:\, S \cap S' = \emptyset} \{ \calE_{S'} \} \bigg] \\
            &\ge \sum_{S' \,:\, S \cap S' = \emptyset} \frac{ \bbP[ \calE_{S'} ]^2 }{ \sum_{S^{\natural} \,:\, S \cap S^{\natural} = \emptyset} \bbP[ \calE_{S'} \cap \calE_{S^{\natural}} ] }, \label{eq:deCaen}
    \end{align}
    where \eqref{eq:deCaen} follows from de Caen's lower bound on the probability of a union \cite{Dec97}.  However, for $S'$ and $S^{\natural}$ both disjoint from $S$ (but possibly overlapping with each other), $\bbP[ \calE_{S'} \cap \calE_{S^{\natural}} ]$ is exactly the same quantity as $\bbE_Q\big[ \calI_{S'}(\bX,\bY) \calI_{S^{\natural}}(\bX,\bY) \big]$ appearing in \eqref{keykey}.  In particular, $\bbP[ \calE_{S'} ]$ corresponds to the case that $S' = S^{\natural}$.  Substituting the expression in \eqref{keykey} gives $\bbP[ \calE_{S'} \cap \calE_{S^{\natural}} ] = 2^{-n(1+\frac{\ell}{k})}$ when $|S' \cap S^{\natural}| = k-\ell$, and substitution into \eqref{eq:deCaen} gives
    \begin{align}
        \bbP[\calI = 0 \,|\, S] 
            &\ge {p \choose k} \frac{ 4^{-n} }{ \sum_{\ell=0}^k {k \choose \ell}{p-k \choose \ell} 2^{-n(1+\frac{\ell}{k})} } \\
            &= \frac{1}{1 + \chi^2(P\|Q)}, \label{equating}
    \end{align}
    where \eqref{equating} follows by equating with \eqref{eq:chi_sq_2}.  From Theorem \ref{verymain1}, we know that $\chi^2(P\|Q) \to 0$ under the conditions of Theorem \ref{thm:identify_dd}, and we conclude that $\bbP[\calI = 0 \,|\,S] \to 1$.  Since this holds regardless of which $S$ is conditioned on, we obtain $\bbP[\calI = 0] \to 1$ as desired.

\section{Berry-Esseen Theorem} \label{sec:berry_esseen}

Our analysis makes use of the following Berry-Esseen theorem, a non-asymptotic form of the central limit theorem.

\begin{theorem} \label{feller}
    {\em (Berry-Esseen Theorem \cite[Theorem 2]{feller})}
    For $j=1,\ldots,p$, let $X_j$ be independent random variables with 
    \begin{align*}
        \mu_j=\bbE[X_j],\quad \sigma_j^2=\var[X_j],~~ \text{\rm and}~~ \rho_j=\bbE[|X_j-\mu_j|^3].
    \end{align*}
    Denote $V=\sum_{j=1}^p \sigma_j^2$ and $T=\sum_{j=1}^p \rho_j$. Then, for any $\lambda\in \bbR$, we have
    \begin{align}
        \bigg|\bbP\bigg[\frac{\sum_{j=1}^p (X_j-\mu_j)}{\sqrt{V}}\geq \lambda\bigg]-\calQ(\lambda)\bigg| \leq  \frac{6T}{V^{3/2}},
    \end{align}
    where $\calQ(t) =\int_t^{\infty} \frac{1}{\sqrt{2\pi}}e^{-\frac{u^2}{2}}\,{\rm d}u$.
\end{theorem}

More precisely, we use the following simple corollary.

\begin{corollary} \label{momlem} 
    Let $Z \sim \mbox{Binomial}(p, q_0)$. Then, for any $\lambda \in \bbR$, the following holds:
    \begin{align}
        \label{eq225eq}
        \bigg|\bbP\bigg[\frac{Z-pq_0}{\sigma\sqrt{p}} \geq \lambda\bigg]-\calQ(\lambda)\bigg| \leq  \frac{6\rho}{\sigma^3 \sqrt{p}},
    \end{align}
    where
    \begin{align}
        \rho &= (1-q_0)^3 q_0 + q_0^3 (1-q_0), \label{eq:rho}\\
        \sigma &=\sqrt{(1-q_0)q_0}.  \label{eq:sigma}
    \end{align}
\end{corollary}
\begin{proof} 
    Since $Z \sim \mbox{Binomial}(p, q_0)$, we can write $Z=\sum_{j=1}^p Z_j$, where the $Z_j$ are i.i.d.~with distribution $\mbox{Bernoulli}(q_0)$.  We shift to a zero-mean summation by writing $Z-pq_0=\sum_{j=1}^p (Z_j-q_0)$, and observe that
    \begin{align}
        \rho_1:=\bbE[|Z_1-q_0|^3]=(1-q_0)^3 q_0 + |0-q_0|^3 (1-q_0)=\rho,
    \end{align}
    and
    \begin{align}
        \sigma_1^2=\bbE[(Z_1-q_0)^2]=(1-q_0)^2 q_0+ (0-q_0)^2(1-q_0)=\sigma^2
    \end{align}
    for $\rho$ and $\sigma$ defined in \eqref{eq:rho}--\eqref{eq:sigma}.  Hence,~\eqref{eq225eq} follows directly from Theorem~\ref{feller} with $T=p\rho_1^3=p \rho^3$ and $V=p \sigma_1^2=p\sigma^2$.
\end{proof}

\bibliographystyle{IEEEtran}
\bibliography{isitbib,JS_References}

\end{document}